\documentclass[a4paper,UKenglish,cleveref,autoref,thm-restate,
runningheads,orivec]{llncs}

\usepackage{balance} 
\usepackage{amsmath,xspace}
\usepackage{etex}
\usepackage{centernot}
\usepackage{stmaryrd}
\usepackage{wrapfig}
\usepackage{amsfonts}
\usepackage{mathtools}
\usepackage[a4paper, margin=2.5cm]{geometry}

\usepackage{xspace}
\usepackage{color}
\definecolor{lgray}{gray}{0.9}
\usepackage{fixltx2e}
\usepackage{fancybox}
\usepackage{ifthen}
\usepackage{pgf}
\usepackage{tikz}
\usetikzlibrary{arrows,automata}
\usepackage{graphics}
\usepackage{subfig}
\usepackage{proof}
\usepackage{hyperref}
\usepackage{soul}
\usepackage{accsupp} 
\usepackage{cancel}
\usepackage{textcomp}
\usepackage{calrsfs}
\usepackage{paralist}
\usepackage{enumitem}

\usepackage{changebar}
\usepackage{todonotes}
\usepackage[super]{nth}

\DeclareMathAlphabet{\pazocal}{OMS}{zplmf}{m}{n}
\newcommand{\mcal}[1]{\pazocal{#1}}

\newcommand{\rulename}[1]{\mbox{\textsc{#1}}}

\newcommand{\qt}[1]{``{#1}"}

\newcommand{\rTo}[1]{\xrightarrow{\ #1}}

\newcommand{\RTO}[2]{\xrightarrow{{#1}}\hspace{-1.2mm}\textsuperscript{${#2}$}}

\newcommand{\trns}[2]{(\RTO{\tau}{\star})_{{#1}}^{{#2}}}

\newcommand{\true}{{\sf true}}
\newcommand{\false}{{\sf false}}

\newlength{\arrow}
\newcounter{sqindex}

 \newcommand{\rom}[1]{ \textup{(\lowercase\expandafter{\romannumeral#1})}}

\usepackage[ruled]{algorithm2e} 

\SetAlFnt{\small}
\SetAlCapFnt{\small}
\SetAlCapNameFnt{\small}
\SetAlCapHSkip{0pt}
\IncMargin{-\parindent}

\newcommand \Until      {{\mathbin{\mcal{U}}\kern-.1em}}
\newcommand \Release     {{\mathbin{\mcal{R}}\kern-.1em}}

\newcommand \Since      {\mathbin{\mcal{S}\kern-.08em}}

\newcommand \X      {{\mathsf{X}}}
\newcommand \g    {{\mathsf{{G}}\kern.08em}}
\newcommand \f    {{\mathsf{{F}}\kern.08em}}
\newcommand \UntilHat   {\mathbin{\LTLhat{\mcal{U}}\kern-.1em}}

\newcommand \SinceHat   {\mathbin{\LTLhat{\mcal{S}}\kern-.08em}}

\newcommand \Not        {\mathopen{\neg}}
\renewcommand \And      {\mathbin{\wedge}}
\newcommand \Or         {\mathbin{\vee}}

\newcommand \ltl        {\textsc{ltl}\xspace}

\newcommand{\set}[1]{\{{#1}\}}
\newcommand{\conf}[1]{\langle{#1}\rangle}

\newcommand{\band}[3]{\bigwedge_{#1}^{#2}{#3}}
\newcommand{\bor}[3]{\bigvee\limits_{#1}^{#2}{#3}}
\newcommand{\bcup}[3]{\bigcup_{#1}^{#2}{#3}}

\newcommand{\tuple}[1]{\left(#1\right)}

\def\<#1>{\mathinner{\langle#1\rangle}}

\newcommand{\chan}{ch}
\newcommand{\schan}{\rulename{ch}} 

\newcommand{\id}{i}

\newcommand{\obsfun}{{\bf x}}
\newcommand{\outfun}{{\bf o}}
\newcommand{\projkx}{\mathsf{\bf proj}^x_k}
\newcommand{\projko}{\mathsf{\bf proj}^o_k}

\newcommand{\projkch}{\mathsf{\bf proj}^{\chan}_k}

\newcommand{\size}[1]{|{#1}|}
\newcommand{\Exp}[1]{2^{#1}}

\newcommand{\msf}[1]{\mathsf{#1}}

\newcommand{\typecvar}{{\scriptstyle@\msf{type}}}
\newcommand{\assigncvar}{{\scriptstyle@\msf{asgn}}}
\newcommand{\readycvar}{{\scriptstyle@\msf{rdy}}}
\newcommand{\lnkcvar}{{\scriptstyle@\msf{lnk}}}

\newcommand{\listen}{\rulename{ls}}

\usepackage{fontawesome}
\usepackage[pdf]{graphviz}

\usepackage[inkscapeformat=pdf]{svg}

\usepackage{apxproof} 

\newtheoremrep{theorem}{Theorem}
\newtheoremrep{lemma}{Lemma}
\newcommand{\BibTeX}{\rm B\kern-.05em{\sc i\kern-.025em b}\kern-.08em\TeX}

\begin{document}
\title{Correct-by-Design Teamwork Plans for Multi-Agent Systems\thanks{This work is funded by
 the Swedish research council
grant: SynTM (No.
2020-03401) (Led by the first author) and the ERC consolidator grant
D-SynMA (No. 772459)(Led by the second author).}
}
\titlerunning{Correct-by-Design Teamwork Plans for Multi-Agent Systems}
%
\author{
	Yehia Abd
	Alrahman
	\and
	Nir Piterman
}
\authorrunning{Y. Abd Alrahman and N. Piterman}
%
\institute{University of Gothenburg, Gothenburg, Sweden\\
	\email{\{yehia.abd.alrahman,nir.piterman\}@gu.se}}
\maketitle              

\begin{abstract}
 We propose \emph{Teamwork Synthesis}, a version of the
     distributed synthesis problem with application to teamwork
     multi-agent systems.
     We reformulate the distributed synthesis question by dropping the
     fixed interaction architecture among agents as input to the
     problem.
     Instead, our synthesis engine tries to realise the goal given the
     initial specifications; otherwise it automatically introduces
     minimal interactions among agents to ensure distribution.
     Thus, teamwork synthesis mitigates a key difficulty in deciding
     algorithmically how agents should interact so that each obtains
     the required information to fulfil its goal.
     We show how to apply teamwork synthesis to provide 
     a distributed solution.

\end{abstract}

\section{Introduction}\label{sec:intro}
Synthesis~\cite{PnueliR89} of correct-by-design multi-agent systems is still one of the most intriguing challenges in the field.
Traditionally, synthesis techniques targeted \emph{Reactive Systems}~--~ systems that maintain continuous interactions with hostile environments.
A synthesis algorithm is used to automatically produce a \emph{monolithic} reactive system that is able to satisfy its goals no matter what the environment does.
Synthesis algorithms have been also extended for other domains, e.g., to support rational environments~\cite{KupfermanS22}, cooperation~\cite{MajumdarPS19,EhlersKB15}, knowledge~\cite{JonesKPL12}, etc.

A major deficiency of traditional synthesis algorithms is that they produce a monolithic program, and thus fail to deal with distribution~\cite{FinkbeinerS05}. In fact, the distributed synthesis problem is undecidable, except for specific configurations~\cite{PnueliR90,FinkbeinerS05}. This is disappointing when the problem we set out to solve is only meaningful in a vibrant distributed domain, such as multi-agent systems.

In this paper, we mount a direct attack on the latter, and especially Teamwork Multi-Agent Systems (or Teamwork MAS)~\cite{NairTM05,PynadathT03}.
Teamwork MAS consist of a set of autonomous agents that share an execution context in which they collaborate to achieve joint goals. They are a natural evolution of reactive systems, where an agent has to additionally collaborate with team members to jointly maintain correct reactions to inputs from the context. Thus, being reactive requires being prepared to respond to inputs coming from the context and interactions from the team.

The context is uncontrolled and can introduce uncertainties for individuals that may disrupt the joint behaviour of the team. For instance, a change in sensor readings  of agent$_k$ that some other agent$_j$ cannot observe, but is required to react to, etc. Thus, maintaining correct (and joint) reactions to contextual changes requires a highly flexible coordination structure~\cite{Tambe97}.
This implies that fixing all interactions within the team in advance is not useful, simply because the required level of connectivity changes dynamically.

Despite that flexible coordination mechanisms are undeniably effective to counter uncertainties, the literature on distributed synthesis and control is primarily focused on fixed coordination, e.g., Distributed synthesis~\cite{PnueliR90,FinkbeinerS05}), Decentralised supervision~\cite{Thistle05,ramadge89}, and Zielonka
synthesis~\cite{Zielonka87,GenestGMW10}. This reality, however, is due
to the fact that there is no canonical model to describe distributed
computations, and hence the focus is on well-known models with fixed
structures.
It is widely agreed that the undecidability result is mainly due to
partial (or lack of) information. The latter can also be rephrased as
\qt{lack of coordination}.
Note that the decidability of a distributed synthesis problem is conditioned on the right match between the given concurrency model and its formulation~\cite{Muscholl15}.

We are left in the middle of these extremes: Distributed synthesis~\cite{PnueliR90,FinkbeinerS05},
Zielonka
synthesis~\cite{Zielonka87,GenestGMW10}, and Decentralised supervision~
\cite{Thistle05}. All are
undecidable except for specific configurations. Zielonka synthesis
 is decidable if synchronising agents are allowed to share their
 entire state, and this produces agents that are exponential
 in the size of the joint deterministic specification.

We propose \emph{Teamwork Synthesis}, a decidable reformulation of the distributed synthesis problem. We reformulate the synthesis question by dropping the fixed interaction architecture among agents. Instead, our approach dynamically introduces minimal interactions when needed to maintain correctness.
Teamwork synthesis consider a set of agent interfaces, an environment model that specifies assumptions on the  context and (possibly) partial interactions among agents, and a formula over the joint goal of the team within the context. A solution for teamwork synthesis is a set of reconfigurable programs, one per agent such that their dynamic  composition satisfies the formula under the environment model.

The contributions in this paper are threefold: \rom{1} we introduce the Shadow transition system (or Shadow TS for short) which distills the essential features of reconfigurable multicast from CTS~\cite{AbdAlrahmanP21}, augments, and disciplines them to support teamwork synthesis; \rom{2} we propose a novel parametric bisimulation that is able to abstract unnecessary interactions, and thus helps producing Shadow TSs with least amount of coordinations, and with size that is, in the worst case, equivalent to the joint deterministic specification. This is a major improvement on the Zielonka approach and with less coordination;  \rom{3}  lastly, we present teamwork synthesis and  show how
     to reduce it to a single-agent synthesis. The
     solution is used to construct an equivalent loosely-coupled distributed one. Our synthesis engine will try realise the goal given the initial
specifications, otherwise it will
automatically introduce additional required interactions among agents
to ensure distributed realisability. Note that those additional
interactions are strategic, i.e., they are introduced dynamically when
needed and disappear otherwise.
Thus, teamwork synthesis will enable us to mitigate a key difficulty in deciding algorithmically how agents should interact so that each obtains the required information to carry out its functionality.


The  paper's structure is as follows: In Sect.~\ref{sec:over}, we give an overview on teamwork synthesis. In Sect.~\ref{sec:bck}, we present a short background materials, and later in Sect.~\ref{sec:scen}, we present a case study to illustrate our approach. In Sect.~\ref{sec:model}, we present the Shadow TS and the corresponding bisimulation. In Sect.~\ref{sec:synth}, we present teamwork synthesis and in Sect.~\ref{sec:conc}, we report our concluding remarks.

\section{Teamwork Synthesis in a nutshell}\label{sec:over}
We consider a team of $K$ autonomous agents that execute in a shared context, and pursue a joint goal. A context can be a physical space or an external entity that may impact the joint goal.

Interaction among team members is established based on a set of channels (or event names), denoted $Y$ and partitioned among all members. An agent, say agent$_k$, can locally control a subset of event names $(Y_k\subseteq Y)$ \emph{by being responsible of sending all messages with channels from $Y_k$} while other agents may be eligible to receive.

We assume that every agent$_k$, partially observes its context by means of reading local sensor observation values $x_k \in X_k$.
Moreover, agent$_k$ may react to new inputs from $X_k$ or messages (with channels  from other agents, i.e., in $(Y\backslash Y_k)$) by generating local actuation signals $o\in O_k$. That is, the signals agent$_k$ uses to control its state, e.g., a robot sends signals to its motor to change direction.

Message exchange is established in a reconfigurable multicast fashion. That is, agent$_k$ may send messages to interested team members, i.e., agents that currently listen to the sending channel. A receiving agent, agent$_j$ for $j\neq k$, can adjust its actuation signals $O_j$  accordingly. Agents can connect/disconnect channels dynamically based on need. An agent  only receives messages on channels that listens to in its current state, and cannot observe others. 

Agent$_k$ starts from a fixed initial state, and in every future execution step it either: observes a new sensor input from $X_k$; receives a message on a channel from $(Y\backslash Y_k)$ that agent$_k$ listens to in the current state;  or sends a message on a channel from $Y_k$ to interested members. In all cases, agent$_k$ may trigger \emph{individual} actuation signals $O_k$ accordingly.

As a team, every team execution starts from a fixed initial state.
Moreover, in every execution step the team either observes an \emph{aggregate sensor input}~--~some members (i.e., a \emph{subset} of $K$) observe an input~--~or exposes a message on channel from $Y$ originated exactly from one member. In both cases, the team may trigger an \emph{aggregate actuation signal} $O$.
Formally, the set of aggregate sensor inputs over $\set{X_k}_{k\in K}$ is $X=\{\obsfun:K \hookrightarrow \bigcup_k X_k ~|~ \obsfun(k) \notin \bigcup_{j\neq k} X_j\ \mbox{and}\ \exists\ k\in K, s.t.\ \obsfun(k)\ \mbox{is defined}\}$. That is, a global observation corresponds to having new sensor values for \emph{some} of the agents. Note that $\obsfun$ is a partial function. Similarly, the set of aggregate actuation signals over $\set{O_k}_{k\in K}$ is $O=\{\outfun:K \hookrightarrow \bigcup_k O_k ~|~ \outfun(k) \notin \bigcup_{j\neq k} O_j\}$. Note that unlike $X$, the set of aggregate output signals $O$ can be empty.

Thus, teamwork synthesis only requires that aggregate observations $X$ and interactions on channels from $Y$ interleave~\cite{pi2} after initialisation $\theta_i$ (i.e., the initial condition), see the assumption automaton $(A)$ below:
\begin{figure}[t]\centering
\includegraphics[scale=.35]{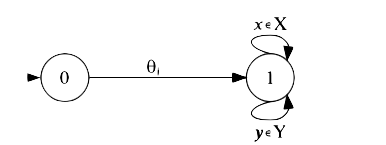}
\vspace{-2.5mm}
\caption{Execution Assumption}\label{fig:execass}
\vspace{-3mm}
\end{figure}

The rationale is that we start from an environment model $E$ that
specifies both aggregate context observations $X$ and (possibly) interactions on
channels from $Y$, i.e., the environment model $E$ may centrally
specify an interaction protocol on channels from $Y$. Then we are given
a set of agent interfaces $\set{\conf{X_k,\ Y_k,\ O_k}}_{k\in K}$ such
that $X$ is the set of aggregate observations over $\set{X_k}_{k\in K}$, 
$O$ is the set of aggregate
actuation signals over $\set{O_k}_{k\in K}$
 as defined before, and $Y=\bcup{k}{}{Y_k}$; and a formula $\varphi$
over the joint goal of the team within $E$ (i.e., the language of
$\varphi$ is in $((Y\cup{X})\times{O})^{\omega}$).

Our synthesis engine will try realise the goal given the initial
protocol description (which can also be empty) on $Y$, and if this is not possible, it will
automatically introduce additional required interactions among agents
to ensure distributed realisability.
We use the Shadow TS, with essential features of reconfigurable multicast, as the underlying distributed model for teamwork synthesis.

Formally, a solution for \emph{teamwork Synthesis} $\mathcal{T}=\langle E\cap A,\ \varphi,\ O\rangle$ is a set of $\size{K}$-Shadow TSs, one for each $\conf{X_k,\ Y_k,\ O_k}$ such that their team composition satisfies $\varphi$ under $E\cap A$, where $E\cap A$ is the standard automata intersection of $E$ and the execution assumption $A$ depicted above.
We show that the teamwork synthesis problem can be reduced to a single-agent synthesis. The solution of the latter can be efficiently decomposed into a set of equivalent shadow TSs.

\section{Background}\label{sec:bck}
We present the background material on symbolic automata for environment's specifications and linear
temporal logic (\ltl).
\begin{definition}[Environment model]\label{def:env}
An environment model $E$ is a deterministic symbolic automaton of the form $E=\langle Q,\ \Sigma,\ \Psi,\ q_0,\ \rho\rangle$,
\begin{itemize}[label={$\bullet$}, topsep=2pt, itemsep=2pt, leftmargin=15pt]
\item $Q$ is a set of states and $q_0\in Q$ is the initial state.
\item $\Sigma$ is a structured alphabet of the form $(Y\cup{X})$.
\item $\Psi$ is a set of predicates over $\Sigma$ such that every predicate $\psi\in\Psi$ is interpreted as follow: $\llbracket\cdot\rrbracket:\Psi\rightarrow {(Y\cup{X})}$.
\item $\rho:Q\times\Psi \rightarrow Q$ is the transition function, s.t. for all transitions $(q,\psi,q'), (q,\psi',q'')\in\rho$,  if $\psi\wedge\psi'$ is satisfiable then $q'=q''$.

\end{itemize}
\end{definition}
The language of $E$, denoted by  $\mathcal{L}_E$, is a set of infinite
sequences of letters in $({Y\cup{X}})^{\omega}$.
Two environment models $E_1$ and $E_2$ can be composed by means of standard automata intersection ($E_1\cap E_2$).

For goal specifications, we use \ltl to specify the goals of individual agents and their joint goals.
We assume an alphabet of the form $(Y\cup{X})\times{O})$ as defined before.
A model $\sigma$ for a formula $\varphi$ is an infinite sequence of letters in $(Y\cup{X})\times{O})$, i.e., it is in $((Y\cup{X})\times{O}))^{\omega}$.
Given a model $\sigma = \sigma_0, \sigma_1, \ldots$, we denote by
$\sigma_i$ the letter at position $i$.

LTL formulas are constructed using the following grammar.
\[
\varphi ::= v\in \bigcup_k (X_k\cup O_k)  ~|~y \in Y~|~
\neg \varphi ~|~ \varphi_1 \Or \varphi_2 ~|~
\X \varphi ~|~
\varphi_1\ \Until\ \varphi_2\
\]

\noindent
For a formula $\varphi$ and a position $i\geq 0$,
$\varphi$ \emph{holds at position $i$ of $\sigma$}, written
$\sigma,i \models \varphi$, where $\sigma_i=(v,\outfun)$, if:
\begin{itemize}[label={$\bullet$}, topsep=2pt, itemsep=2pt, leftmargin=15pt]
\item
For $x_k \in X_k$ we have $\sigma,i \models x_k$ iff
$v\in X$ and and $v(k)=x_k$.
That is, $x_k$ is
satisfied if $v(k)$ is defined and equal to $x_k$.\footnote{
	It is possible to say $v(k)$ is defined and \emph{not} equal to $x_k$
	by $\bigvee_{x_k\neq x\in X_k} x$.
}
\item
For $o_k \in O_k$ we have $\sigma,i \models o_k$
iff $\outfun(k)=o_k$
\item
For $y \in Y$ we have $\sigma,i \models y$ iff $v=y$
\item $\sigma,i \models \neg \varphi $ iff $\sigma,i \not\models
\varphi$
\item $\sigma,i \models \varphi \Or \psi$ iff $\sigma,i \models
  \varphi$ or $\sigma,i \models \psi$
\item $\sigma,i \models\ \X \varphi$ iff $\sigma,i+1 \models
  \varphi$
\item $\sigma,i \models \varphi\ \Until\ \psi$ iff there exists $k\geq i$
  such that $\sigma,k \models \psi$ and $\sigma,j \models \varphi$
  for all $j$, $i \leq j < k$
\end{itemize}
If $\sigma,0 \models \varphi$, then $\varphi$
\emph{holds} on $\sigma$ (written $\sigma \models \varphi$). A
set of models $M$ satisfies $\varphi$, denoted $M \models \varphi$, if
every model in $M$ satisfies $\varphi$.
A formula is \emph{satisfiable} if the set of models satisfying it is
not empty.

We use the usual abbreviations of the Boolean connectives
$\And$, $\rightarrow$, and $\leftrightarrow$ and the usual definitions
for $\true$ and $\false$.
We introduce the following temporal abbreviations $\f\phi=\true\
\Until\ \phi$,
$\g\phi=\neg F \neg \phi$, and
$\phi_1\Release\phi_2=\neg (\neg \phi_1 \Until \neg \phi_2)$.
\section{Distributed Product Line Scenario}\label{sec:scen}
We use a distributed product line scenario to illustrate \emph{Teamwork Synthesis} and its underlying principles.

The product line, in our scenario, is operated by three robot arms: \rom{1} the $\rulename{tray}$ arm that observes inputs on the input-tray and forwards them for processing; \rom{2} the $\rulename{proc}$ arm that is responsible for processing the inputs; \rom{3} and the $\rulename{pkg}$ arm that packages and delivers the final product.

The operator of the product line is an \emph{uncontrollable} human, adding inputs, denoted by $(\msf{in})$, to the input-tray. The operator serves as the execution context in which the three robot arms operate. Only the $\rulename{tray}$ arm can observe the input $(\msf{in})$.

The specifications of the robot arms are as follows:
The interface of the \rulename{tray} is of the form $\msf{Int}_{\rulename{tray}}=\conf{\set{\msf{in}},\set{\msf{fwd}},\set{\msf{rFwd}}}$. That is, the {\rulename{tray} arm can observe the input $\msf{in}$ on the input-tray, it can also send a message on channel $\msf{fwd}$, and it has one actuation signal $\msf{rFwd}$ to instruct its motor to get ready to forward the input. The {\rulename{f}-automaton below specifies its part of the interaction protocol.

{\qquad
\includegraphics[scale=.4]{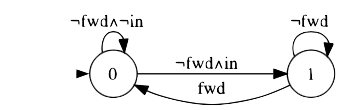}}

That is, the {\rulename{tray} arm can forward by sending a message on $\msf{fwd}$ only after it observes an input $\msf{in}$. The safety goals of the {\rulename{tray} are:\smallskip

$
\begin{array}{l}
\varphi_1=\g \msf{ (in    \rightarrow  rFwd)}\ \&\\[4pt]
\qquad\ \g \msf{((rFwd \And (\X \Not fwd)) \rightarrow (\X rFwd))}\\[4pt]
\end{array}
$

That is, the motor gets ready to forward whenever an input is observed. Moreover, the motor remains ready to forward as long as forwarding did not happen.

The interface of the {\rulename{proc} arm is of the form  $\msf{Int}_{\rulename{proc}}=\break\conf{\emptyset,\set{\msf{proc}},\set{\msf{rProc}}}$ . That is, the {\rulename{proc} arm cannot observe any input, but it can send a message on $\msf{proc}$, and it has one actuation signal $\msf{rProc}$ to instruct its motor to get ready to process the input. The {\rulename{p}-automaton below and the \ltl formula $\rulename{pd}$ specify the arm part in the interaction protocol.

{\qquad
\includegraphics[scale=.4]{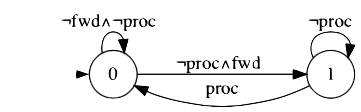}}

\[\rulename{pd}=\g \msf{((proc \And (\X in) \And (\X \X fwd))} \rightarrow\msf{\X (dlv\ R \Not proc))}\]

Namely, the {\rulename{proc} arm can process by sending a message on $\msf{proc}$ only after a forward has happened. Moreover, the arm cannot process twice in row without a deliver in between. We will use $\msf{A}(\rulename{pd})$ to denote the automaton representing \rulename{pd}.

The safety goals of the {\rulename{proc} arm are as follows:\smallskip

$
\begin{array}{l}
\varphi_2=\g \msf{(fwd \rightarrow  rProc)}\ \&\\[4pt]
\qquad\ \g \msf{((rProc \And (\X \Not proc)) \rightarrow (\X rProc))}\\[4pt]
\end{array}
$

That is, the motor gets ready to process whenever forward happens. Moreover, the motor remains ready to process as long as processing did not happen.

 The interface of the {\rulename{pkg} arm is  $\msf{Int}_{\rulename{pkg}}=\conf{\emptyset,\set{\msf{dlv}},\set{\msf{rDlv}}}$ . That is, the {\rulename{pkg} arm cannot observe any input, but it can deliver by sending a message on $\msf{dlv}$, and it has one actuation signal $\msf{rDlv}$ to instruct its motor to get ready to package and deliver the input. The {\rulename{d}-automaton below specifies its part of the interaction protocol.

{\qquad
\includegraphics[scale=.4]{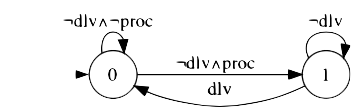}}

The {\rulename{pkg} arm can send a message on $\msf{dlv}$ only after processing has happened. The safety and liveness goals of the {\rulename{pkg} arm are:\smallskip

$
\begin{array}{l}
 \varphi_3=\g \msf{(proc \rightarrow  rDlv)}\ \&\\[4pt]
\qquad\ \g \msf{((rDlv \And (\X \Not dlv)) \rightarrow (\X rDlv))}\ \\[4pt]
\end{array}
$

That is, the motor gets ready to deliver  whenever process happens. Moreover, the motor remains ready to deliver as long as delivering did not happen. We also require $\g \f \msf{(rDlv)}$, i.e., the motor must also be ready to delivering infinitely often.

We have the following assumption on the operator $\rulename{op}$:

 {\qquad
\includegraphics[scale=.4]{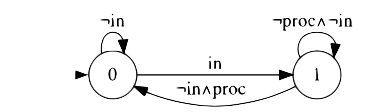}}

Namely, after a first input the operator waits for processing to happen before it puts a new input.

Finally, we require $\g \f \msf{(in)}$, i.e., the operator must supply input infinitely often.

We assume that all signals are initially off. That is:\smallskip

$\begin{array}{c}
\theta=\neg\msf{in}\wedge\neg\msf{fwd}\wedge\neg\msf{proc}\wedge\neg\msf{dlv}\wedge\neg\msf{rFwd}\wedge\neg\msf{rProc}\wedge\neg\msf{rDlv}\\[4pt]
\end{array}
$\smallskip

Notice that these specifications are written from a central point of
view. For instance, the formula $\rulename{pd}$ of the  \rulename{proc}
arm predicates on ($\msf{in},\ \msf{fwd},\ $ and $\msf{dlv}$) even if
it cannot observe them. To be able to enforce this formula, we need to
be able to automatically introduce strategic and minimal interactions
among agents at run-time, only when needed (!), and this is the role of
teamwork synthesis.

The instance of teamwork synthesis $\mathcal{T}=\langle E\cap A,\ \varphi,\ O\rangle$ is:
\[
\begin{array}{cl}
\rom{1} & E=\rulename{f}\cap\rulename{p}\cap\rulename{d}\cap \msf{A}(\rulename{pd})\cap\rulename{op}\\[4pt]
\rom{2} & A=\mbox{is an instance of the automaton depicted in Fig.~\ref{fig:execass}}\\[4pt]
\rom{3} &  \varphi=\theta\wedge\varphi_1\wedge\varphi_2\wedge\varphi_3\wedge (\g\f(\msf{in})\rightarrow\g\f(\msf{rDlv}))\\[4pt]
\end{array}\]
\[
\begin{array}{cl}
\rom{4} & O= \left \{
\begin{array}{c}
	\emptyset,
	\set{(\rulename{tray}\mapsto \msf{rFwd})},\\
	\set{(\rulename{proc}\mapsto \msf{rProc})},
	\set{(\rulename{pkg}\mapsto \msf{rDlv})},\\
	\set{(\rulename{tray}\mapsto \msf{rFwd}), (\rulename{proc}\mapsto \msf{rProc})},\\
	\set{(\rulename{tray}\mapsto \msf{rFwd}), (\rulename{pkg}\mapsto \msf{rDlv})},\\
	\set{(\rulename{proc}\mapsto \msf{rProc}), (\rulename{pkg}\mapsto \msf{rDlv})},\\
	\set{(\rulename{tray}\mapsto \msf{rFwd}), (\rulename{proc}\mapsto \msf{rProc}),(\rulename{pkg}\mapsto \msf{rDlv})}
\end{array}
\right	\}\\[2ex]
\end{array}
\]
\smallskip

A solution for $\mathcal{T}=\langle E\cap A,\ \varphi,\ O\rangle$ is a $3$-Shadow TSs, one for each $\conf{X_k,\ Y_k,\ O_k}_{k\in\set{1,2,3}}$ such that $T_1\|T_2\| T_3\models\varphi$ under $E\cap A$.

We will revisit the scenario, at the end of Sect.~\ref{sec:synth}, to show the distributed realisation of this problem and its features.
\section{Shadow Transition Systems}\label{sec:model}

We formally present the \emph{Shadow Transition System} and we use it to define the behaviour of individual agents. We also define how to compose different agents to form a team.


\begin{definition}[Shadow TS]\label{def:shadow} A shadow TS is of the form ${T_k} = \langle S_k,\ \rulename{Int}_k,\ \rulename{Act}_k, \Delta^k_e, \ \Delta_k,\ L_k,\ \listen^k,\ s^k_0\rangle$, where:
\begin{itemize}[label={$\bullet$}, topsep=2pt, itemsep=2pt, leftmargin=10pt]
\item $S_k$ is the set of states of $T_k$ and $s^k_0\in S_k$ its
initial state.

\item $\rulename{Int}_k=\langle X_k,\ \schan_k,\ O_k\rangle$ is the  interface of $T_k$, where
\begin{itemize}
\item $X_k$ is an observation alphabet, $\schan_k$ is a set of interaction channels, and $O_k$ is an output (or actuation) alphabet. We use $\id$ to range over elements in $X_k$ or $\schan_k$;
\item $\listen^k: S_k\rightarrow\Exp{\schan_k}$ is a channel listening function. That is, $\listen^k$ defines (per state) the channels that $T_k$ listens to.
\end{itemize}

\item $\rulename{Act}_k\subseteq(\schan_k\times\set{!,?}\times\Upsilon_k)$ is the set of messages. Intuitively, a message consists of a channel  $\chan\in\schan_k$, a type (send $!$ or receive $?$), and a  load (or contents) $\upsilon\in\Upsilon_k$.

\item $L_k:
S_k\rightarrow(\schan_k\cup({X_k}\uplus\set{\bot}))\times{(O_k\uplus\set{\bot})}$
 is a labelling function where $\bot$ denotes undefined label, i.e.,
$L_k$ labels states with input (output) letters that were observed
(correspondingly produced).

\item $\Delta^k_e\subseteq S_k\times(\schan_k\cup{X_k})$ denotes the
environment potential moves from $S_k$, i.e., $\Delta^k_e$ can be
thought of as a ghost transition relation denoting the instantaneous
perception of $T_k$ of its environment.

\item $\Delta_k\subseteq S_k\times\rulename{Act}_k\times S_k$ is the transition relation of $T_k$. The relation $\Delta_k$ can be thought of as a shadow transition relation of $\Delta^k_e$. That is, for every potential move in $\Delta^k_e$, there must be a corresponding shadow transition in $\Delta_k$ as follows:

\begin{itemize}

\item For every state $s\in S_k$ and every  letter $i\in(\schan_k\cup{X_k})$, if $(s,i)\in \Delta^k_e$ then there exists $o\in{O_k}, s'\in S_k$ such that
$ L_k(s')=(\id,o)\ \text{and}\ (s,a,s')\in \Delta_k\ \text{for some}\ a\in\rulename{Act}_k$

\item For any state $s\in S_k$, if for every letter $i\in(\schan_k\cup{X_k})$,\ $(s,i)\notin \Delta^k_e$ and there exists $o\in{O_k}, s'\in S_k$ such that $(s,a,s')\in \Delta_k\ \text{for some}\  a\in\rulename{Act}_k$ then $a$ must be a \emph{receive}.

\end{itemize}
\end{itemize}
\end{definition}


Shadow TSs can be composed to form a team as in Def.~\ref{def:comp} below. We use $\projkx$ to denote the projection of a team label into $X_k$ of agent$_k$, and similarly for  $\projkch$ and $\projko$, i.e., for projection on $\schan_k$ and $O_K$ respectively. We use $\bot$ to denote that the projection is undefined.

\begin{definition}[Team]\label{def:comp}
Given a set $K=\set{1,\dots,n}$ of shadow TSs  ${T_k} = \langle S_k,\ \rulename{Int}_k,\ \rulename{Act}_k, \Delta^k_e, \ \Delta_k,\ L_k,\ \listen^k,\ s^k_0\rangle$
where $k\in K$, their composition
$\|_k T_k$ is the team
${T} = \langle S,\ \rulename{Int},\ $\rulename{Act}$,\ \Delta_e, \ \Delta,\ L,\ \rulename{ls},\ s_0\rangle$,

\begin{itemize}[label={$\bullet$}, topsep=2pt, itemsep=2pt, leftmargin=10pt]
\item
  $S = (s_1,\dots,s_n)$, \,
  $s_0 = (s_0^1,\dots,s_0^n)$, \, $\rulename{Act} =
  \bcup{k}{}{\rulename{Act}_k}$, \, $\Upsilon = \bcup{k}{}{\Upsilon_k}$

\item $\rulename{Int}=\conf{X,\schan,O}$ such that $\schan = \bcup{k}{}{\schan_k}$,

  $X=\{\obsfun:K \hookrightarrow \bigcup_k X_k ~|~ \obsfun(k) \notin \bigcup_{j\neq k} X_j\}$,

  and   $O=\{\outfun:K \hookrightarrow \bigcup_k O_k ~|~ \outfun(k) \notin \bigcup_{j\neq k} O_j\}$

\item[]\hspace{-3mm} $\Delta =
\left \{\left (
\begin{array}{c}
	(s_1,\ldots, s_n),\\
	(c,!,\upsilon),\\
	(s'_1,\ldots, s'_n)
\end{array} \right )
\left |\hspace{-1mm}
\begin{array}{l}
 \exists k\in\set{1,n}~.~(s_k,(c,!,\upsilon),s'_k)\in \Delta_k \mbox{ and } \forall
 j\neq k .\\
	(1)~ (s_j,(c,?,\upsilon),s'_j) \in \Delta_j \mbox{ and }
	c\in\listen^j(s_j)
	\mbox { or}\\
	(2)~ c\notin \listen^j(s_j) \mbox { and } s'_j=s_j \\
\end{array}\hspace{-3mm}
\right .
\right	\}
$\\[1ex]

\item[]\hspace{-3mm}$\Delta_e =
\left \{\left (
\begin{array}{c}
	(s_1,\ldots, s_n),\id

\end{array} \right )
\left |~
\begin{array}{l}
((s_1,\ldots, s_n),(c,!,\upsilon),(s'_1,\ldots, s'_n))\in \Delta \\
\mbox{for some } (c,!,\upsilon)\in\rulename{Act}\
\mbox{ such that  }\\
 L((s'_1,\ldots, s'_n))=(i,o)\ \mbox{for some}\ o\in {O}
\end{array}
\right .
\right	\}
$\\[1ex]

\item Given a state $(s_1,\ldots, s_n)$, let $x=\bigcup_k\{k\mapsto\projkx(L_k(s_k))~|~$ $\projkx(L_k(s_k))\neq\bot\}$,\  $\chan=\bigcup_k\projkch(L_k(s_k))$, and \
$o=\bigcup_k\{k\mapsto\projko(L_k(s_k))~|~$ $\projko(L_k(s_k))\neq\bot\}$,
then  $L((s_1,\ldots, s_n)) = (x,o)$ if $x\neq\emptyset$ and $(\chan,o)$ otherwise.
In the systems we construct we achieve that $\chan$ is always a unique
value in $\schan$.
\item
  $\listen((s_1,\ldots, s_n)) = \bcup{k}{}{\listen^k(s_k)}$
\end{itemize}
\end{definition}

Note that the composition in Def.~\ref{def:comp} does not necessarily produce a shadow TS. However, our synthesis engine will generate a set of shadow TSs such that their composition is also a shadow TS.

Intuitively, multicast channels are blocking. That is, if there exists an agent$_k$ with a send transition $(s_k,(c,!,v),s'_k)\in\Delta_k$ on channel $c$ then every other parallel agent$_j,s.t. j\neq k$ (that listens to $c$ in its current state, i.e., $c\in\listen(s_j)$) must supply a matching receive transition $(s_j,(c,?,v),s'_j)\in\Delta_j$ or otherwise the sender is blocked. Other parallel agents that do not listen to $c$ simply cannot observe the interaction, and thus cannot block it. We restrict attention to the set of shadow TSs $\mathcal{T}$ that satisfy the following property:
\begin{property}[Local broadcast]\label{pr:brd}
$\forall T\in\mathcal{T}, s\in S_T,\ \text{we have that} 
c\in\listen(s)\ \text{\bf iff}\ (s,(c,?,\upsilon),s') \in \Delta_T$ for
all $c\in \schan$ and $\upsilon\in\Upsilon$.
\end{property}
Thus, a shadow TS cannot block a message send by listening to its channel and not supplying a corresponding receive transition. This reduces the semantics to asynchronous local broadcast. That is, message sending cannot be blocked, and is sent on local broadcast channels rather than a unique public channel ($\star$) as in CTS~\cite{AbdAlrahmanP21}.

A run of $T$ is the infinite sequence $r=s_0a_0s_1a_1s_2\dots$ such that for all $k\geq 0: (s_k,a_k,s_{k+1})\in \Delta$ and $s_0$ is the initial state. 
An execution of $T$ is the projection of a run $r$ to \emph{state labels}. That is, for a run  $r=s_0a_0s_1a_1s_2\dots$, there is an execution $w$ induced by $r$ such that $w=L(s_0)L(s_1)L(s_2)\dots$.
We use $\mathcal{L}_T$ to denote the language of $T$, i.e., the set of all executions of $T$. For a specification $\mathsf{spec}\subseteq (({X}\cup\schan)\times{O})^{\omega}$, we say that $T$ satisfies $\mathsf{spec}$ if and only if $\mathcal{L}_T\subseteq\mathsf{spec}$. Note that the key idea in our work is that we use a specification that only refers to \emph{aggregate input and output}, and is totally insensitive to messages. As we will see later, the latter will be used by a synthesis engine to ensure distributed realisability.

\begin{lemma} The composition operator is a commutative monoid.
\end{lemma}
\begin{proof}
The proof follows directly by Property~\ref{pr:brd} and the definition of $\|$. There, the existential and universal quantifications on $k$ are insensitive to the location of $s_k$ in $(s_1,\ldots, s_n)$ for $k\in\set{1,n}$. A sink state, denoted by $0$, (i.e., a state with zero outgoing transitions and empty listening function, i.e., $\listen(0)=\emptyset$) is the $\mathsf{Id}$-element of $\|$ because it cannot influence the composition.
\end{proof}


We define a notion of parameterised bisimulation that we use to
efficiently decompose a Shadow TS.
\paragraph{\bf $\mathcal{E}$-Bisimulation}\label{sec:bisim}
Consider the TS  $T$ with a finite state space $S=\set{s_1,\dots,s_n}$ that is composed
with the TS $\mathcal{C}$ (that we call the \emph{parameter TS}). The
latter has a finite state space
$\mathcal{E}=\set{\epsilon_1,\dots,\epsilon_m}$ and will be used as the
basis to minimise the former. That is, TS $\mathcal{C}$ is only agent
that $T$ can interact with.
This is the only bisimulation used in this paper.
When we write bisimulation we mean parameterised
bisimulation.
We first introduce some notations:
\begin{itemize}[label={$\bullet$}, topsep=2pt, itemsep=2pt, leftmargin=10pt]
\item
$(\epsilon\rTo{a}\epsilon')$: a parameter state $\epsilon$ permits
message $a=(c,!,\upsilon)$ iff $\epsilon$ can receive $a$, does not
listen to $c$, or  is the sender. Formally,\smallskip

$
\begin{array}{l}
	\rom{1}  \quad c\in\listen(\epsilon)\ \mbox{and}\ (\epsilon,(c,?,\upsilon), \epsilon')\in\Delta_{\mathcal{E}}\ \mbox{or}\\
	\rom{2}  \quad c\not\in\listen(\epsilon)\ \mbox{and}\ (\epsilon=\epsilon')\ \mbox{or}\\
	\rom{3}  \quad (\epsilon,(c,!,\upsilon), \epsilon')\in\Delta_{\mathcal{E}}
\end{array}
$ \medskip

Note that this item and {Property~\ref{pr:brd}} ensure that message send is autonomous and cannot be restricted by the parameter TS.
\item
$(s\rTo{a!}s')$ : $s$ sends message $a=(c,!,\upsilon)$ iff $(s,(c,!,\upsilon), s')\in\Delta$
\item
$(s\rTo{a?}s')$ :  $s$ \emph{receives}
message $a=(c,!,\upsilon)$ and updates iff  $c\in\listen(s)$,
$(s,(c,?,\upsilon), s')\in\Delta$, $L(s)\neq L(s')$. 

\item
$(s\rTo{\tau_a}{}s')$ :  $s$ can
\emph{discard} 
iff $a=(c,!,\upsilon)$,  $c\in\listen(s)$,\ $(s,(c,?,\upsilon),
s')\in\Delta$, $L(s)=L(s')$. Note the state's label did not change by receiving.  We drop
the name $a$ from  $\tau_a$ when $a$ is arbitrary. 

Note that all kinds of receives ($\rTo{\tau_a}\ $  or  $\ \rTo{a?}$) cannot happen without a joint message-send.

\item We use $\trns{\epsilon}{\epsilon'}$ to denote a sequence (possibly empty) of arbitrary discards ($\rTo{\tau_a}$ for any $a$), starting when the parameter state is $\epsilon$ and ending with $\epsilon'$.
We define a family of transitive closures as the minimal relations
satisfying: \rom{1}
$s\ \trns{\epsilon}{\epsilon'}\ s$; and  
\rom{2} if $s_1\ \trns{\epsilon_1}{\epsilon_2}\ s_2$, $(s_2\rTo{\tau_a}{}s_3)$, $(\epsilon_2\rTo{a}\epsilon_3)$, and $s_3\ \trns{\epsilon_3}{\epsilon_4}\ s_4$ then $s_1\ \trns{\epsilon_1}{\epsilon_4}\ s_4$.

These are the reflexive and transitive closure of
$\rTo{\tau}$ while making sure that also the
parameter supplies the sends that are required.

%
%
\item
We will use $(s\rTo{a})$ when $s$ has $a$ transition, and
$(s\not\rTo{a})$ when $s$ has no $a$ transitions.

\end{itemize}

\begin{definition}[$\mathcal{E}$-Bisimulation]\label{def:bisim}
Let the shadow TS $\mathcal{C}$ with finite state space $\mathcal{E}=\set{\epsilon_1,\dots,\epsilon_m}$ be a parameter TS. An $\mathcal{E}$-bisimulation relation $\mathcal{R}$ is a \emph{symmetric} $\mathcal{E}$-indexed family of relations $\mathcal{R}_{\epsilon}\subseteq S\times S$ for $\epsilon\in\mathcal{E}$, such that whenever $(s_1,s_2)\in\mathcal{R}_{\epsilon}$ then $L(s_1) = L(s_2)$, and
for all $a\in\rulename{Act},$ if $\epsilon\rTo{{a}}\epsilon'$ then

\begin{enumerate}
\item $\ s_1\rTo{a!} s'_1$  \quad implies\quad  $\exists s'_2,$  $\ s_2\rTo{a!}s'_2$ and $(s'_1,s'_2)\in\mathcal{R}_{\epsilon'}$;

\item $\ s_1\rTo{a?} s'_1$\quad implies\quad $s_2\not\rTo{\tau_a}$\quad and
\[ \left (
\begin{array}{cc}
	\mbox{if} & s_2\rTo{a?}\ \ \mbox{then}\ \ \exists s'_2,\ \ s_2\ \rTo{a?}\  s'_2\  \ \mbox{and}\ \ (s'_1,s'_2)\in\mathcal{R}_{\epsilon'}\\
	\mbox{else} & \exists s'_2,s''_2, \epsilon'',\  \ s_2\
	(\RTO{\tau}{\star})_{\epsilon''}^{\epsilon}\ s''_2\rTo{a?}\ s'_2,\ 
	\mbox{and}\  \ (s'_1,s'_2)\in\mathcal{R}_{\epsilon'}
\end{array} \right )
 \]

\item $\ s_1\rTo{\tau_a} s'_1$\quad implies\quad $s_2\not\rTo{a?}$\quad and
\[ \left (
\begin{array}{cc}
	\mbox{if} & s_2\rTo{\tau_a}\ \ \mbox{then}\ \ \exists s'_2,\ \ s_2\ \rTo{\tau_a}\  s'_2\  \ \mbox{and}\ \ (s'_1,s'_2)\in\mathcal{R}_{\epsilon'}\\
	\mbox{else} & \exists s'_2,\epsilon'',\  \ s_2\
	(\RTO{\tau}{\star})_{\epsilon''}^{\epsilon}\  s'_2\ \  \mbox{and}\
	(s'_1,s'_2)\in\mathcal{R}_{\epsilon'}
\end{array} \right )
 \]
\end{enumerate}

Two states $s_1$ and $s_2$ are equivalent with respect to a parameter state $\epsilon\in\mathcal{E}$, written $s_1\sim_{\epsilon} s_2$, iff there exists an $\mathcal{E}$-bisimulatin $\mathcal{R}$ such that $(s_1,s_2)\in\mathcal{R}_{\epsilon}$. Please note that $\mathcal{R}$ is symmetric.
\end{definition}

Def.~\ref{def:bisim} equates two states with same labelling with respect to the current parameter state $\epsilon$ if: (1) they supply the same send transitions; (2) they supply same receive transitions or one can discard a number of messages  and reach a state in which it can supply a matching receive; (3) is similar to (2) except for the \qt{else} part where one state can supply an arbitrary number of discard (possibly none). In all cases, both states are required to evolve  to equivalent states under the \emph{next} parameter state $\epsilon'$.

Note that case $2$ and $3$ (and their symmetrics) in
Def.~\ref{def:bisim} allow an agent to avoid participating in
interactions that do not affect it.

We use $\sim_0$ to denote the equivalence under the empty parameter $\mathcal{O}$. That is, $\mathcal{O}$ has a singleton sink state $0$. The parallel composition $\|$ in Def.~\ref{def:comp} is a commutative monoid, and $0$ is the $\mathsf{id}$-element. Thus, we have that $(s,0)$ is equivalent to $s$ for all $s\in S$.

We need to prove that $\sim_{\epsilon}$ is closed under the parallel
composition in Def.~\ref{def:comp} within a composite parameter
$\mathcal{C}$. That is, a parameter $\mathcal{C}$ of the form
$\mathcal{C}_1\|\dots\|\mathcal{C}_n$ for some $n$. For a composite
state $\epsilon=(\epsilon_1,\dots,\epsilon_n)$ and $w=\{i_1,\ldots,
i_j\}$, we use $(\epsilon_{i_1},\dots,\epsilon_{i_j})$ to denote a
$w$-cut of $\epsilon$. That is, a projection of
$\epsilon$ on states $(\epsilon_{i_1},\dots,\epsilon_{i_j})$.
Moreover, we use ${\epsilon\backslash p}$ to denote $\epsilon$ without cut $p$.

\begin{theorem}[$\sim_{\epsilon}$ is closed under $\|$] For all states $s_1,s_2$ of a shadow TS, all composite parameter states $\epsilon\in\mathcal{E}$ of the form $\epsilon=(\epsilon_1,\dots,\epsilon_n)$, and all cuts $p$ of length $w\leq n$, we have that:

$\quad s_1\sim_{\epsilon}s_2$\ implies\ $(s_1,p)\sim_{\epsilon\backslash p}(s_2,p)$
\end{theorem}
\begin{proof}
It is sufficient to prove that for every composite parameter state $\epsilon$, the following relation:
\[\mathcal{R}_{\epsilon\backslash p}=\set{((s_1,p),(s_2,p))~|~\mbox{for all states}\ s_1,s_2,\ s.t.\ (s_1\sim_{\epsilon}s_2)}\] is a $(\epsilon\backslash p)$-bisimulation.

Recall that $\|$ is a commutative monoid, and thus it is closed under commutativity, associativity, and $\mathsf{Id}$-element. Thus, the rest of the proof is by induction on the length $(w)$ of the projection with respect to the history of the parameter TS. The key idea of the proof is that send actions of the form $(c,!,\upsilon)$ can only originate from within the composition, i.e., can be sent by $s_1$ (or $s_2$) or $\epsilon$. Moreover, a receive action of the form $(c,?,\upsilon)$ can only happen jointly with a corresponding send while the latter is autonomous.
\end{proof}


\section{Teamwork Synthesis}\label{sec:synth}

Given an environment model $E$ that specifies both \emph{aggregate} context observations $X$ and scheduled interactions on channels from $Y$, the execution assumption $A$ automaton depicted in Fig.~\ref{fig:execass}, a formula $\varphi$ over the joint goal of the team within $E$ (i.e., $\mathcal{L}(\varphi)\subseteq ((Y\cup{X})\times{O})^{\omega}$), a set of agent interfaces $\set{\conf{X_k,\ Y_k,\ O_k}}_{k\in K}$ such that $X$ is the set of \emph{aggregate} observations of $\set{X_k}_{k\in K}$, $O$ is the set of \emph{aggregate} actuation signals of $\set{O_k}_{k\in K}$ as defined in Sect.~\ref{sec:over}, and $Y=\bcup{k}{}{Y_k}$,
a solution for \emph{teamwork Synthesis} $\mathcal{T}=\langle E\cap A,\ \varphi,\ O\rangle$ is a set of $\size{K}$ of Shadow TSs, one for each $\conf{X_k,\ Y_k,\ O_k}$ such that $T_1\|\dots\| T_k\models\varphi$ under $E\cap A$.

We show that the teamwork synthesis problem can be reduced to a single-agent synthesis. The solution of the latter can be efficiently decomposed into a set of loosely coupled shadow TSs, where their composition is an equivalent implementation.

\begin{theorem} Teamwork Synthesis whose specification $\varphi$ is a $\mathsf{\rulename{gr(1)}}$ formula~\cite{nsyn} of the form $\theta\wedge\g \phi\wedge (\band{\id=1}{n}{\g\f\ \lambda_i}\rightarrow \band{\id=1}{m}{\g\f\ \gamma_i})$ can be solved with effort $\mathcal{O}(m.n.(\size{E\cap A}.\size{O})^2)$, where $\size{E\cap A}$ is the number of transitions in ${E\cap A}$.
\end{theorem}

\begin{proof}
We construct $\hat{E}=\langle \hat{Q},\ \hat{\Sigma},\ \hat{\Psi},\ \hat{q_0},\ \hat{\rho}\rangle$ that extends $E\cap A=\langle Q,\ \Sigma,\ \Psi,\ q_0,\ \rho\rangle$ to include the set of \emph{aggregate} signals in $O$. To simplify the notations, we freely use $o\in{O}$ to mean the predicate that characterises it.
The components of $\hat{E}$ in relation to $E\cap A$ are:
\begin{itemize}[label={$\bullet$}, topsep=2pt, itemsep=2pt, leftmargin=15pt]
\item $\ \hat{Q}=Q$,\quad  $\hat{q_0}={q_0}$,\quad $\hat{\Sigma}= \Sigma\times O$
\item We extend the interpretation function $\llbracket\cdot\rrbracket$ to include variables in $O$. That is,  $\llbracket\cdot\rrbracket:\hat{\Psi}\rightarrow (Y\cup{X})\times{O}$
\item
$\hat{\rho}=$

$\quad
\begin{array}{l}
\left \{\left (
\begin{array}{c}
	(q_0,\theta\wedge\theta_i,q)
\end{array} \right )
\left |~
\begin{array}{l}
 (q_0,\theta_i,q)\in\rho
\end{array}
\right .
\right	\}\cup\\[2ex]
\left \{\left (
\begin{array}{c}
	(q,\psi\wedge o,q')
\end{array} \right )
\left |~
\begin{array}{l}
 q\neq q_0,\ (q,\psi,q')\in\rho\ \mbox{and}\ o\in{O}
\end{array}
\right .
\right	\}
\end{array}
$
\end{itemize}
We use the construction above to construct a symbolic fairness-free $\mathsf{\rulename{ds}}$~\cite{nsyn}.
We transform $\hat{E}$ into an equivalent fairness-free symbolic discrete system $\mathsf{\rulename{ds}}$ $\mathcal{D}=\conf{V_d,\rho_d,\theta_d}$ in the obvious way. We use the variables $X'$ where $\size{X'}=\sum_k\lceil\log\size{ X_k}\rceil$ to encode observations, the variables $Y'$ where $\size{Y'}=\log\size{Y}$ to encode channels, the variables $O'$ where $\size{O'}=\sum_k\lceil\log\size{ O_k}\rceil$ to encode outputs, and the variable $\mathsf{st}$\footnote{To simplify the notation we will consider $\mathsf{st}$ to be a non-boolean variable.} to encode the states $\hat{Q}$. That is, $\mathcal{D}=\conf{V_d,\rho_d,\ \theta_d}$, where:

\begin{itemize}[label={$\bullet$}, topsep=2pt, itemsep=2pt, leftmargin=15pt]

\item $V_d=(X'\cup Y'\cup O'\cup\set{\mathsf{st}})$
\item We define $\rho_d(V_d,{V_d}')$ which is a predicate on the current assignment to $V_d$ in relation to the next assignment. We use the primed copy ${V_d}'$ to refer to the next assignment of $V_d$.

$\rho_d=\bor{(q_1,\psi_1,q_2),\ (q_2,\psi_2,q_3)\in {\rho'}}{}{{\begin{array}{c}
\psi_1\wedge(\psi_2)'\wedge (\mathsf{st}=q_1)\wedge(\mathsf{st}'=q_2)
\end{array}}}$
\item $\theta_d=\theta\wedge (\mathsf{st}=q_0)$
\end{itemize}

For a state $(s,q)\in\Exp{(X'\cup Y'\cup O')}\times Q$ where $Q$ is the domain of $\msf{st}$, we say that $(s,q)\models v$ iff $v\in s$. We naturally generalise satisfaction to boolean combination of $(V_d)$ and $({V_d}')$.
Now, given a Teamwork Synthesis problem $\mathcal{T}=\langle \mathcal{D}(\hat{E}),\ \varphi,\ O\rangle$, where $\mathcal{D}=\conf{V_d,\rho_d,\ \theta_d}$, $\varphi$ is a $\mathsf{\rulename{gr(1)}}$ formula divided into a liveness assumption $\band{\id=1}{n}{\g\f\ \lambda_i}$,\  a liveness goal $\band{\id=1}{m}{\g\f\ \gamma_i}$, and a safety goal $\g \phi$, we construct a $\mathsf{\rulename{gr(1)}}$ game $G=\langle V,\ X,\ O,\ \theta_e,\ \theta_s,\ \rho_e,\ \rho_s,\ \varphi_g\rangle$ as follows:

We use $\phi$ in the safety goal to prune any transition in $\rho_d$ with condition on $O$ that is in conflict with $\phi$ ( it is unsafe). That is, we construct ${\rho'_d}$ from ${\rho_d}$ by removing all transitions $t\in{\rho_d}$ such that $\neg (t\rightarrow\phi)$. The initial transitions from state $q_0$ are not subject to check against $\phi$. Clearly, $(\X {\rho'_d})\rightarrow\g \phi$. Moreover, $\rho'_d$ also encodes the environment safety by definition.
Now, our game is as follows: $G=\langle V_d,\ (X'\cup Y'),\ (O'\cup\set{\mathsf{st}}), \true, \theta_d,\ \true,\ \rho'_d,\ \varphi_g\rangle$, where $\varphi_g=\band{\id=1}{n}{\g\f\ \lambda_i}\rightarrow\band{\id=1}{m}{\g\f\ \gamma_i}$
\end{proof}
To support response formulas of the form $\g(x\rightarrow\f y)$ and general \rulename{ltl} safety formulas instead, the complexity is adjusted as follows: $\mathcal{O}((m+g).(n+a).(\size{E\cap A}.\size{O}.\size{\varphi(s)}.\Exp{(a+g)})^2)$, where $m$, $n$ are adjusted by adding the number of response assumptions $a$ and response guarantees $g$ while $\size{\varphi(s)}$ is the size of the safety goal. This is because the disciplined environment model $E\cap A$ will be intersected with the safety goal and each response formula. A response formula can be encoded in a two-state automaton~\cite{PitermanPS06}.

The solution of the $\mathsf{\rulename{gr(1)}}$ game can be used to construct a Mealy machine with interface $\conf{X,\ Y,\ O}$ as defined below:

\begin{definition}[Mealy Machine] A Mealy machine $M$ is of the form $M=\langle Q,\ I,\ O,\ q_0,\ \delta\rangle$, where:
\begin{itemize}[label={$\bullet$}, topsep=0pt, itemsep=0pt, leftmargin=10pt]
\item $Q$ is the set of states of $M$ and $q_0\in Q$ is the  initial state.

\item $I=(X\cup Y)$ is an alphabet, partitioned into a set of \emph{aggregate} sensor inputs $X$ and a set of channels $Y$, and $O$ is the \emph{aggregate} output alphabet.

\item $\delta: Q\times({Y}\cup{X})\rightarrow Q\times{O}$ is the transition function of $M$.

\end{itemize}
\end{definition}
The language of $M$, denoted by $\mathcal{L}_M$, is the set of infinite sequences $(({Y}\cup {X})\times {O})^{\omega}$ that $M$ generates.

We will use $M$ to build a corresponding shadow TS.
Namely, we construct a language equivalent shadow TS $T$ with set of states $S=\delta$ as shown below. Note that the constructed mealy machine in the previous step has exactly \emph{one outgoing transition from the initial state} $q_0$, i.e., $\size{(q_0,(i,o),q)\in\delta}=1$. This is ensured by the execution assumption $A$ automaton depicted in Fig.~\ref{fig:execass}.

\begin{lemma}[From Mealy to Shadow TS] Given the constructed Mealy machine $M=\langle Q,\ I,\ O,\ q_0,\ \delta\rangle$ then we use a function $f:\delta\rightarrow \rulename{Act}$, a load $\Upsilon$, and channels $\schan$ to construct a shadow TS ${T} = \langle S,\ \rulename{Int},\ $\rulename{Act}$,\ \Delta_e, \ \Delta,\ L,\ \listen,\ s_0\rangle$ with $\size{\delta}$ many states, and $\mathcal{L}_T=\mathcal{L}_M$.
\label{lem:mtots}
\end{lemma}
\begin{proof}
We construct  $T$ as follows:
\begin{itemize}[label={$\bullet$}, topsep=4pt, itemsep=2pt, leftmargin=10pt]
\item $S=\set{(q,(i,o),q') ~|~ (q,(i,o),q')\in\delta}$

\item $s_0=(q_0,(i,o),q)$ for the unique $q\in Q$, s.t. $(q_0,(i,o),q)\in\delta$

\item $L((q,(i,o),q'))=(\id,o)$

\item $\rulename{Int}=\conf{X, \schan, O}$ where $X=\{\obsfun:K \hookrightarrow \bigcup_k X_k ~|~\ \obsfun(k) \notin \bigcup_{j\neq k} X_j\}$;  $\schan\subseteq(\Exp{K})\backslash\set{\emptyset}\cup Y$, i.e., the maximal set of channels that agent may use to interact, where $K$ is the set of agent identities; and $O=\{\outfun:K \hookrightarrow \bigcup_k O_k ~|~\ \outfun(k) \notin \bigcup_{j\neq k} O_j\}$


\item $\listen(s)=\emptyset$ for all $s\in S$ and $\rulename{Act}\subseteq(\schan\times\set{!}\times\Upsilon)$ where $\Upsilon\subseteq X\cup\set{\emptyset}$. Note that \rulename{Act} is restricted to send messages.

\item $\Delta = $\\[.1ex]

 $\left \{\left (
\begin{array}{c}
	(q,(i,o),q'),\\
	a,\\
	(q',(i',o'),q'')
\end{array} \right )
\left |
\begin{array}{l}
 (q,(i,o),q')\in\delta,(q',(i',o'),q'')\in\delta\\ \mbox{and}\ a=f((q',(i',o'),q''))
\end{array}
\right .
\right	\}
$\\

\item We use $\projkx(\id)$ to project $\id$ on $X_k$, s.t. function $f$ is:
\[f((s,(i,o),s'))=
\begin{cases}
      \tuple{\id,!,\emptyset} & \mbox{if}\  \id\in Y\\
      \tuple{{\mathsf{ids}},!,\id} &  \mbox{if}\ \id\in {X},\ \mathsf{ids}=\set{k ~|~  \projkx(\id)\neq\bot} \\

   \end{cases}
\]

\item $\Delta_e= 
\left \{\left (
\begin{array}{c}
	(q,
	(i,o),
	q'), i'
\end{array} \right )
\left |
\begin{array}{l}
 (q,(i,o),q'),a,(q',(i',o'),q''))\in \Delta
\end{array}
\right .
\right	\}
$

It is not hard to see that $\mathcal{L}_T=\mathcal{L}_M$.
\end{itemize}

\end{proof}

%
%
%
%

\begin{lemma}[Decomposition]\label{lem:decomp}
A shadow TS   ${T} = \langle S,\ \rulename{Int},\ \rulename{Act}, \Delta_e, \ \Delta,\ L,\ \listen,\ s_0\rangle$, as constructed in Lemma \ref{lem:mtots}, can be decomposed into a set of  TSs $\set{T_k}_k$ for $k\in \{1,n\}$, $s.t.\ T\sim_0 (T_1\|\dots\| T_n)$.
\end{lemma}
\begin{proof} We construct the components of each $T_k$ as follows:
\begin{itemize}[label={$\bullet$}, topsep=2pt, itemsep=2pt, leftmargin=15pt]
\item
  $S_k = {S}$, \,
  $s_0^k = { s_0}$

\item For each $s\in S$, $(s,(c,!,\upsilon),s')\in \Delta$, and each $T_k$, we have that
\begin{enumerate}
\item if  $k\in c$ and $|c|=1$ then  $(s,(c,!,\upsilon),s')\in \Delta_k$
\item if  $k\in c$ and $|c|>1$ then  $(s,(c,!,\upsilon),s')\in \Delta_k$ and $(s,(c,?,\upsilon),s')\in \Delta_k$
\item if  $y= c$ for some $y\in Y_k$  then  $(s,(c,!,\upsilon),s')\in \Delta_k$
\item otherwise  $(s,(c,?,\upsilon),s')\in \Delta_k$
\end{enumerate}

\item $\Delta^k_e =
\{(s,\id) ~|~
(s,(c,!,\upsilon),s')\in \Delta_k,\, L_k(s')=(i,o)\}
$
 \item $\rulename{Act}_k = \set{a~|~(s,
	a,
	s')\in\Delta_k}$
\item
  $L_k(s) = \mathsf{\bf proj}_k(L(s))$, i.e., the projection of $L(s)$ on Agent$_k$
\item
  $\listen^k(s) = \set{c~|~(s,
	(c,?,\upsilon),
	s')\in\Delta_k}$
\end{itemize}
It is sufficient to prove that $s_0$ of $T$ is equivalent to $s'_0$ of the composition $(T'=T_1\|\dots\| T_n)$ under the empty parameter state $0$, and that each $T_k$ is indeed a shadow TS. The key idea of the proof is that the construction creates isomorphic shadow TSs that are fully synchronous. That is, they have the same states and transition structure, and only differ in state labelling, listening function, and transition role (send $!$ or receive $?$). Thus, every send  in $\Delta$ of $T$ is divided into a set of  one (or more sends) $!$ and exactly $(n-1)$-receives $?$. In case of more than one send (as in item $(2)$), then for each $T_k$ that implements such send there must be a matching receive  with same source and target states. By Def.~\ref{def:comp}, we can reconstruct $T'$ that is isomorphic to $T$ with $\listen(s)=\schan$ for all $s\in S$. However, we only need to prove $T\sim_0 T'$ under the $\mathcal{O}$-parameter which cannot interact, and thus $\listen(s)$ is not important after constructing $T'$.
\end{proof}

Lemma~\ref{lem:decomp} provides an upper bound on number of
communications each agent $T_k$ must participate in within the team. We
will use the {$\mathcal{E}$-Bisimulation} in Sect.~\ref{sec:bisim} to
reduce such number with respect to the rest of the composition (or
team). That is, given a set $K=\set{1,\dots,n}$ of agents and for each
agent$_k$, we minimise the corresponding $T_k$ with respect to the
parameter $T=\|_{j\in(K\backslash \set{k})}T_j$. Note that
Lemma~\ref{lem:decomp} produces agents that are of the same size of the
deterministic specification.
Size wise, we recall that in Zielonka
synthesis~\cite{Zielonka87,GenestGMW10})
agents are exponential in the size of the
deterministic specification.
We will reduce this size even more by constructing the quotient shadow TS of $T_k$:
\begin{definition}[Quotient Shadow TS]\label{def:quo}
For a shadow TS ${T_k} = \langle S_k,\ \rulename{Int}_k,\
\rulename{Act}_k, \Delta^k_e, \ \Delta_k,\ L_k,\ \listen^k,\
s^k_0\rangle$, $\mathcal{E}$ the state-space of
$\displaystyle\mathop{\parallel}_{j\neq
k}T_j$, and $\mathcal{E}$-bisimulation, the quotient shadow TS
$[T_k]_{\sim}=\langle S'_k,\rulename{Int}'_k$,$\rulename{Act}'_k$,
$\Delta^{k'}_e, \Delta'_k,L_k,\listen^{k'},s^{k'}_0\rangle$
\begin{itemize}[label={$\bullet$}, topsep=4pt, itemsep=2pt, leftmargin=15pt]
\item $S'_k=\set{s_{\sim}~|~ s\in S_k}$ with $s_{\sim}=\set{s'\in S_k~|~ s\sim_e s'\ \mbox{for }\break e\in\mathcal{E} }$

\item $s^{k'}_0={s^k_0}_{\sim}$

\item $\Delta'_k=$

$\qquad
\begin{array}{l}
\left \{\left (
\begin{array}{c}
	s_{\sim},(c,!,\upsilon),s'_{\sim}
\end{array} \right )
\left |~
\begin{array}{l}
 (s,(c,!,\upsilon),s')\in\Delta_k
\end{array}
\right .
\right	\}\cup\\[2ex]

\left \{\left (
\begin{array}{c}
	s_{\sim},(c,?,\upsilon),s'_{\sim}
\end{array} \right )
\left |~
\begin{array}{l}
s_{\sim}\neq s'_{\sim},\
 (s,(c,?,\upsilon),s')\in\Delta_k
\end{array}
\right .
\right	\}
\end{array}
$
\item $\Delta^{k'}_e =
\{(s,\id) ~|~
(s,(c,!,\upsilon),s')\in \Delta'_k,\, L_k(s')=(i,o)\}
$
 \item $\rulename{Act}'_k = \set{a~|~(s,
	a,
	s')\in\Delta'_k}$
\item
  $\listen^{k'}(s) = \set{c~|~(s,
	(c,?,\upsilon),
	s')\in\Delta'_k}$
\end{itemize}
\end{definition}

The definition of {$\mathcal{E}$-Bisimulation} in Sect.~\ref{sec:bisim} can be easily converted to an algorithm (see~\cite{paigetrajan}) that efficiently (almost linear) computes {$\mathcal{E}$-Bisimulation} as the largest fixed point.

\subsubsection*{\bf\large Scenario Revisited}\label{sec:revisit}
\begin{figure}[t]
\begin{tabular}{l}
$
\begin{array}{l}\\
\mbox{\rom{1} $T_1$: The Tray Arm Agent}\\[2ex]
\includegraphics[scale=.4]{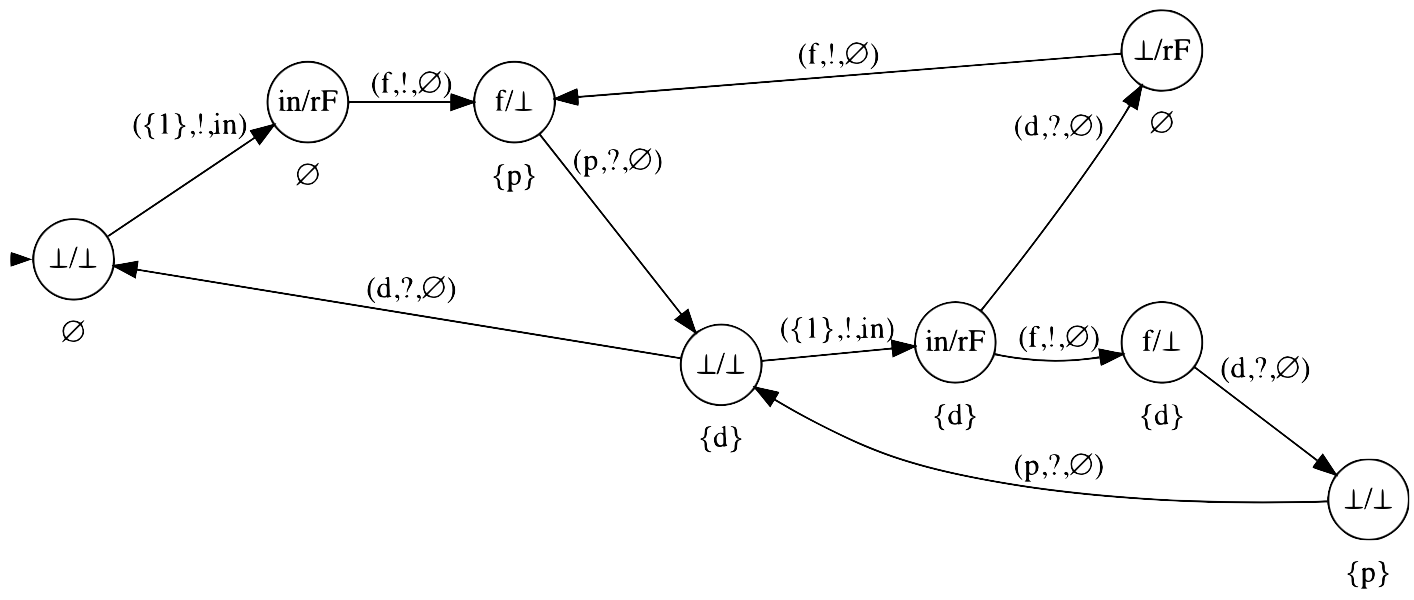}\\
\end{array}$ \qquad

$
\begin{array}{l}
\mbox{\rom{2} $T_2$: The processing Arm Agent}\\[2ex]
  \includegraphics[scale=.3]{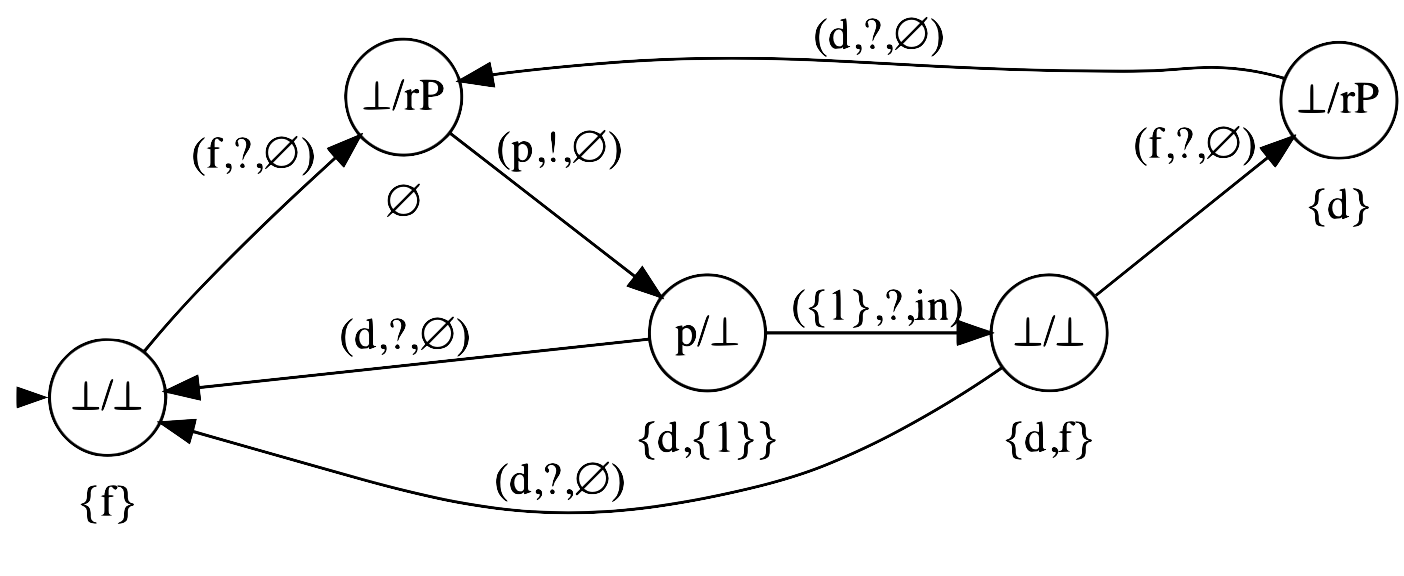}\\
\end{array}$ \\

$
\begin{array}{l}
\mbox{\rom{3} $T_3$: The Packaging Arm Agent}\\[2ex]
  \includegraphics[scale=.3]{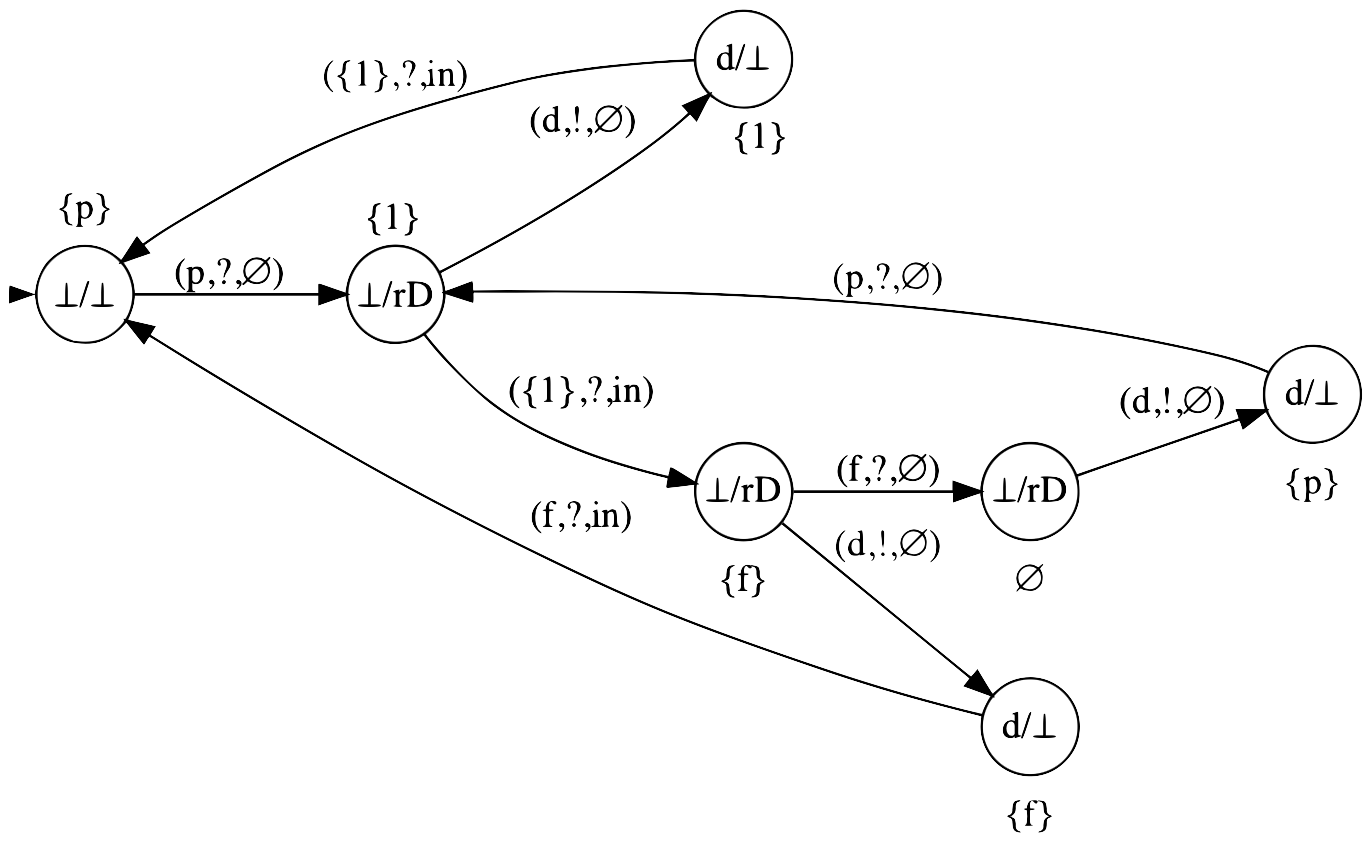}

\end{array}$

\end{tabular}
\caption{The distributed Realisation of the Product line}
\label{fig:drlise}
\end{figure}

The distributed realisation of the teamwork synthesis instance, in Sect.~\ref{sec:scen}, is depicted in Fig.~\ref{fig:drlise}, where each arm is supplied with a shadow TS that represents its correct behaviour. For a shortcut, we only use the \qt{first letter} of a channel name $y\in Y_k$ and the \qt{first two letters} of an output letter name $o\in O_k$ in the figures, e.g., we use the shortcuts $\mathsf{f}$ for $\mathsf{fwd}$, $\mathsf{rF}$ for $\mathsf{rFwd}$, etc.

 Note that every state of $T_k$ for $k\in\set{1,2,3}$ is labelled with input/output letter and a set of channels that $T_k$ listens to in this state, e.g., $T_1$ in Fig.~\ref{fig:drlise}\rom{1} initially has empty input/output letter (the initial state is labelled with $\bot/\bot$) and is not listening to any channel (the listening function is initially $\emptyset$). Moreover, $T_2$ and $T_3$ in Fig.~\ref{fig:drlise}\rom{1} and \rom{2} are initially listening to channel $\mathsf{f}$ and respectively $\mathsf{p}$.

 Transitions are labelled with either message send $(!)$ or receive $(?)$. For instance, $T_1$ can initially send the message $(\set{1},!,\msf{in})$ on channel $\set{1}$ \emph{independently and move alone} to the next state in which $T_1$ reads the letter $(\msf{in})$ from the input-tray, and consequently signals its motor to get ready to forward, i.e., by signalling $(\msf{rF})$. Recall that this is a shadow transition of the potential trigger of the environment from the initial state. This is akin to say that once $T_1$ senses a trigger from the environment, it immediately permits it by providing a shadowing transition. Note that $T_2$ and $T_3$ are initially busy waiting to receive a message on $\msf{f}$ and respectively on $\msf{p}$ to kick start their executions.

Clearly, message $(\set{1},!,\msf{in})$ on channel $\set{1}$ is a strategic interaction that our synthesis engine added to ensure distributed realisability. Notice that $T_2$ and $T_3$ do not initially listen to $\set{1}$ and they cannot observe the interaction on it, but later they will listen to it when they need (e.g., after $(\msf{p},!,\emptyset)$).

By composition, as defined in Def.~\ref{def:comp}, initially $T_1$ move independently, and from the next state $T_1$ sends the message $(\msf{f},!,\emptyset)$ in which $T_2$ participates while $T_3$ stays disconnected. Indeed, $T_3$ only gets involved in the third step. Notice how the listening functions of these TSs change dynamically during execution, and allowing for loosely coupled distributed implementation. The latter has a feature that in every execution step one can send a message, and the others are either involved (i.e., they receive) or they cannot observe it (i.e., they do not listen).

Recall that state labels are the elements of executions and the transition labels are complimented by the synthesis engine to ensure distributed realisability. As one can see, all TSs initially start from states that satisfy the initial condition $\theta$. Indeed, there is no signal initially enabled. Moreover, the composite labelling of states in future execution steps satisfies the formula $\varphi$ under the environment model $E$ and the execution assumption $A$.

Note that the machines in Fig.~\ref{fig:drlise} is everything we need. That is, unlike supervisory control~\cite{ramadge89} where the centralised controller is finally composed with the environment model, and the composition is checked against the goal, we do not  have  such requirement. Indeed, the machines in Fig.~\ref{fig:drlise} fully distribute the control.

The results in this paper are unique, and aspire to unlock distributed synthesis for multi-agent systems for the first time.

\section{Concluding Remarks}\label{sec:conc}
We introduced teamwork synthesis which reformulates the original distributed synthesis problem ~\cite{PnueliR90,FinkbeinerS05} and casts it on teamwork multi-agent systems. Our synthesis technique relies on a flexible coordination model, named Shadow TS, that allow agents to co-exist and interact based on need, and thus limits the interaction to interested agents (or agents that require information to proceed). 

Unlike the existing distributed synthesis problems, our  formulation is decidable, and can be reduced to a single-agent synthesis. We efficiently decompose the solution of the latter and minimise it for individual agents using a novel notion of parametric bisimulation. We minimise both the state space and the set of interactions each agent requires to fulfil its goals. The rationale behind teamwork synthesis is that we reformulate the original synthesis question by dropping the fixed interaction architecture among agents as input to the problem. Instead, our synthesis engine tries to realise the goal given the initial specifications; otherwise it automatically introduces minimal interactions among agents to ensure distributed realisability. Teamwork synthesis shows algorithmically how agents should interact so that each is well-informed and fulfils its goal.

\paragraph{\bf Related works} We report on related works with regards to concurrency models used for distributed synthesis, bisimulation relations, and also other formulations of distributed synthesis.

Shadow TS adopts the reconfigurable semantics approach from CTS~\cite{AbdAlrahmanP21,rchk,rcp}, but it is actually weaker in terms of synchronisation. Indeed, the requirement in Property~\ref{pr:brd} lifts out the blocking nature of multicast, and thus the semantics of the Shadow TS is reduced to a \emph{local broadcast}, (cf.~\cite{info19,scp20,forte18,DABL14,DAlG20,forte16,AMP22}). That is, message sending can no longer be blocked, and is broadcasted on local channels rather than a unique public channel $\star$ (broadcast to all) as in  CTS~\cite{AbdAlrahmanP21}. The advantage is that the semantics of Shadow TS is asynchronous and no agent can force other agents to wait for it. It is definitely weaker than shared memory models as in~\cite{PnueliR90,FinkbeinerS05} and it is also weaker than the synchronous automata of supervision~\cite{ramadge89} and Zielonka automata~\cite{Zielonka87,GenestGMW10}. Note that the last two adopt the multi-way synchronisation (or blocking rendezvous) of Hoare's CSP calculus~\cite{Hoare21a}. Thus, the synchronisation dependencies are lifted out in our model. Intuitively, an agent can, at most, block itself to wait for a message from another agent, but in no way can block the executions of others unwillingly.

Our notion of bisimulation in Def.~\ref{def:bisim} is novel with respect to existing literatures on bisimulation~\cite{CastellaniH89,MilnerS92,Sangiori93}. To the best of our knowledge, it is the only bisimulation that is able to abstract actual messages, and thus reduce synchronisations. It treats receive transitions in a sophisticated way that allows it to judge when a receive or a discard transition can be abstracted safely. It has a branching nature like in~\cite{GlabbeekW96}, but is stronger because the former cannot distinguish different $\tau$ transitions. It is parametric like in~\cite{Larsen87}, but is weaker in that it can abstract actual receive transitions.

When it comes to distributed synthesis, there is a plethora of formulations. Here, we only relate to the ones that consider hostile environments. These are: Distributed synthesis~\cite{PnueliR90,FinkbeinerS05},
Zielonka
synthesis~\cite{Zielonka87,GenestGMW10}, and Decentralised supervision~
\cite{Thistle05}. Unlike teamwork synthesis,  all are, in general, 
undecidable except for specific configurations (mostly with a tower of exponentials~\cite{KupfermanV01,MadhusudanT01}). Zielonka synthesis
 is decidable if synchronising agents are allowed to share their 
 entire state, and this produces agents that are exponential 
 in the size of the joint deterministic specification. Teamwork synthesis produces agents that are, in the worst case, the size of the joint deterministic specification.

\paragraph{\bf Future works}
We want to generalise the execution assumption $A$ of teamwork synthesis depicted in Fig.~\ref{fig:execass} to a more balanced scheduling between interaction events $Y$ and context  events $X$, inspired by RTC control~\cite{ABDPU21}. That is, we want to provide a more relaxed built-in transfer of control between the interaction and the context events. The latter would majorly simplify writing specifications. We want also to extend the Shadow TS to allow multithreaded agents, and thus eliminates interaction among co-located threads. 

Clearly, the positive results in this paper makes it feasible to provide tool support for Teamwork synthesis, and with a more user-friendly interface.


%
%
%
\bibliographystyle{splncs04}
\bibliography{biblio}
\appendix

\end{document}